\documentclass[twoside]{article}

\usepackage[accepted]{aistats2021}

\usepackage{amsfonts}
\usepackage{amsmath}
\usepackage{amsthm}
\usepackage{todonotes}
\usepackage{listings}
\usepackage{graphicx}
\usepackage{tikz}
\usepackage{wrapfig}
\usepackage{algorithm}
\usepackage{algorithmic}

\usepackage[round]{natbib}

\newcommand{\eps}{\varepsilon}
\newcommand{\M}{\mathcal{M}}
\newcommand{\R}{\mathbb{R}}
\newcommand{\pbad}{fp_\mathrm{bad}}
\newcommand{\pgood}{fp_\mathrm{good}}
\newcommand{\neigh}{G^\delta_{V+}}
\newcommand{\neighp}{G^\delta_{V'+}}

\newtheorem{theorem}{Theorem}
\newtheorem{defn}{Definition}
\newtheorem{lemma}{Lemma}

\DeclareMathOperator{\Rat}{Rat}
\DeclareMathOperator{\Summ}{Sum}
\DeclareMathOperator{\Prodd}{Prod}

\begin{document}

\twocolumn[

\aistatstitle{The Base Measure Problem and its Solution}

\aistatsauthor{ Alexey Radul \And Boris Alexeev }

\aistatsaddress{ Google Reseach \And Google Research } ]

\begin{abstract}
  Probabilistic programming systems generally compute with probability
  density functions, leaving the base measure of each such function
  implicit.  This mostly works, but creates problems when densities
  with respect to different base measures are accidentally combined or
  compared.  Mistakes also happen when computing volume corrections
  for continuous changes of variables, which in general depend on the
  support measure.  We motivate and clarify the problem in the context
  of a composable library of probability distributions and bijective
  transformations.  We solve the problem by standardizing on
  Hausdorff measure as a base, and deriving formulas for comparing and
  combining mixed-dimension densities, as well as updating densities
  with respect to Hausdorff measure under diffeomorphic
  transformations.  We also propose a software architecture that
  implements these formulas efficiently in the
  common case.  We hope that by adopting our solution, probabilistic
  programming systems can become more robust and general, and make a
  broader class of models accessible to practitioners.
\end{abstract}


\section{Introduction}

Suppose we are designing a composable library of software to represent
probability distributions and transformations thereof which,
following \cite{tfp-repo}, we will call ``bijectors''.  TensorFlow
Probability \citep{tfp-repo} and PyTorch Distributions \citep{paszke2019pytorch}
are such libraries in their own right, targeting the probabilistic
machine learning space.  General-purpose probabilistic programming
languages such as Stan~\citep{carpenter2017stan},
BLOG~\citep{li2013blog}, Anglican~\citep{toplin-et-al-anglican}, or Venture~\citep{mansinghka2014venture} must
of necessity also include such libraries, to implement their primitive
distributions and deterministic functions.

Suppose furthermore that we are designing this library to operate on
explicitly vector-valued probability distributions.\footnote{No
  distinction need be made for our purposes between vectors, matrices,
  tuples, or other structures, as long as joint distributions over
  non-trivial powers of $\R$ are in scope.}  This is the case for
TensorFlow Probability and PyTorch Distributions, for instance, to
take advantage of vectorized hardware; and studying the
general case is instructive even for a scalar design.

In this setting, it's conventional to represent a probability
distribution $P$ as a function $p: \R^n \to \R$ computing the
probability density of $P$ at points in $\R^n$, together with a
sampler $x \sim s()$ drawing random variates $x \in \R^n$ distributed
according to $P$.  Given a map $f : \R^n \to \R^n$, the pushforward
$fP$ is then sampled as $x' \sim f(s())$, and its density is computed
as
\begin{equation}
  fp(x') = \frac{1}{|\det(Jf_x)|}p(x),\ \ \mathrm{for}\ x = f^{-1}(x'). \label{eq:naive-density}
\end{equation}
The Jacobian-determinant correction $1/|\det(Jf_x)|$ accounts for the
possibility that $f$ changes the volume of an infinitesimal volume
element near $x$.  The Jacobian determinant can be computed by forming
the Jacobian of $f$, for example with automatic differentiation; but
for many functions $f$, it's available more efficiently.
Thus a conventional choice is to package such $f$, together with their
inverse $f^{-1}$, in a Bijector class with a method for computing said
Jacobian determinant.

This conventional architecture admits a serious bug.  We name this bug
the \emph{Base Measure Problem}, because it consists of neglecting the
base measure with respect to which we are computing densities.  Our
contributions answer these questions:
\begin{itemize}
\item What's the problem?  We give a clear and intuitive example of
  the Base Measure Problem in Section~\ref{sec:example};
\item How common is it?  We briefly survey several areas where the Base
  Measure Problem recurs in different guises in Section~\ref{sec:occurrences};
\item What's the right answer?  We propose a more nuanced standard base
  measure in Section~\ref{sec:solution}, and derive complete density
  manipulation formulas for that choice using standard results; and
\item How do we fix our software?  We detail common special cases of
  our formulas in Section~\ref{sec:software}, as well as how
  to arrange a software library to optimize them.
\end{itemize}

We stress that while we do propose to explicitly represent information
about measures, all the information we will need will be local to a
single point (the point itself and various directional derivatives thereat),
and thus require no symbolic algebra to compute with.

\section{Motivating Example}
\label{sec:example}

Consider the uniform distribution $P$ on the unit
circle in $\R^2$.  The natural probability density function to write
down for this is

\begin{equation*}
  p(x, y) =
  \begin{cases}
    \frac{1}{2\pi} &\mathrm{when}\ x^2 + y^2 = 1, \\
    0 &\mathrm{otherwise}.
  \end{cases}
\end{equation*}

Of course, if we were sticklers we would note that the base measure
implied by the type of the samples is Lebesgue measure on $\R^2$; and
with respect to this base measure the density of $P$ is $+\infty$ on
points on the unit circle.  But that's clearly less helpful to our
users than $1/2\pi$, and $1/2\pi$ is \emph{a} density for $P$, just
with respect to\footnote{Technically, the density function is the
  Radon--Nikodym derivative of our probability measure with
  respect to the base measure.} Lebesgue measure along the circle.  So
let's wing it and go with that.

\begin{figure}
  \centering
  \includegraphics[width=0.48\textwidth,trim=40 10 60 10,clip]{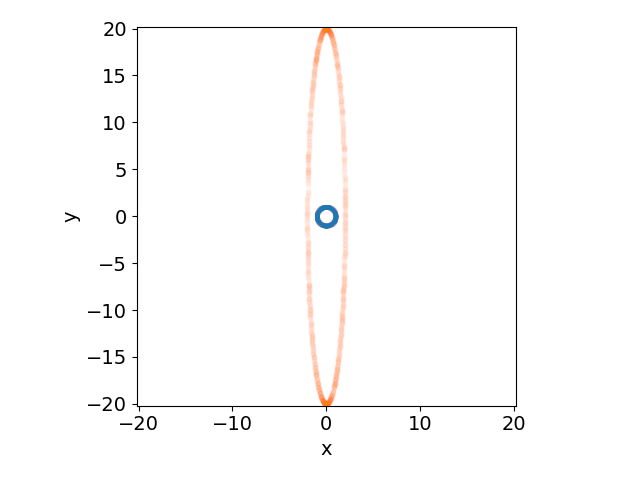}
  \caption{3000 samples from the uniform distribution on the unit
    circle (blue), and the same samples scaled non-isotropically by 2x
    along the $x$ axis and 20x along the $y$ axis (orange).  The
    transformed distribution is clearly not uniform, even though the
    original distribution is, and the Jacobian of the transformation
    is constant over $\R^2$.}
  \label{fig:ellipse}
\end{figure}

Now, a user looking at a density computed for a sample has no way to
know which base measure we meant, so we have created an instance of
the base measure problem.  To see how it bites us,
let $f$ be a somewhat contrived non-isotropic scaling of $\R^2$,
given by $f(x, y) = (2x, 20y)$.  What happens when we try to
compose our uniform distribution $P$ with our non-isotropic scaling
$f$?  The sampler is fine, but the conventional density
rule~(\ref{eq:naive-density}) gives
\begin{align*}
  \pbad(x', y') & = \frac{1}{40}p(f^{-1}(x', y')) \\
  & =
    \begin{cases}
      \frac{1}{80\pi} &\mathrm{when}\ (x'/2)^2 + (y'/20)^2 = 1, \\
      0 &\mathrm{otherwise}.
    \end{cases}
\end{align*}

Now we have definitely made a mistake.  First of all, $\pbad$ doesn't
integrate to 1, because the perimeter of the ellipse $(x'/2)^2 +
(y'/20)^2 = 1$ is approximately $81.28$, which is considerably less
than $80\pi\approx 251.32$.  Second, as we can see by drawing a few
samples and plotting them in Figure~\ref{fig:ellipse}, the true
distribution $fp$ isn't uniform!  It's
clearly denser near $(0, 20)$ than $(2, 0)$.\footnote{With respect to arc length as the base measure.}

\subsection{What went wrong?}

The problem is that when computing the volume change induced by
the change of variables $f$, we forgot that the base measure for $P$
wasn't Lebesgue on $\R^2$.  It's actually Lebesgue
along the circle, and $f$ changes arc length differently at different
points.  Indeed, locally near the point $(x, y)$, the circle is the
line in the direction $(-y, x)$.  The directional derivative of $f$ is
$(-2y, 20x)$, and the change in arclength that $f$ induces is
therefore $\sqrt{4y^2 + 400x^2}$.  Hence the correct density is
\begin{align*}
  \pgood(x', y') =\ & 1/2\pi\sqrt{y'^2/100 + 100x'^2} \\
  &\mathrm{when}\ (x'/2)^2 + (y'/20)^2 = 1,\ \mathrm{and} \\
  &0\ \mathrm{otherwise}.
\end{align*}

\section{How common is this problem?}
\label{sec:occurrences}

While we chose to introduce the base measure problem on a continuous example, it
actually occurs most often when transforming discrete distributions
embedded in $\R^n$.  The density function of, say, the Bernoulli
distribution is $1/2$ at $0$ or $1$ and $0$ elsewhere,
but this is a density with respect to counting measure, not Lebesgue
measure on $\R$.  Therefore, when transforming this distribution with
a bijector, we should not apply any Jacobian correction, because all
bijections are counting-volume-preserving.

The same problem also shows up when dealing with distributions on
symmetric matrices, such as the Wishart or LKJ distributions.  The $k
\times k$ symmetric matrices occupy a lower-dimensional sub-manifold
of $\R^{k \times k}$, and the density of the Wishart distribution is
with respect to Lebesgue measure on that submanifold rather than all
of $\R^{k \times k}$.  Bijections will in general change the volume
under that measure differently than the volume under Lebesgue measure
on $\R^{k \times k}$.

Sometimes changes of variables occur in the inference algorithm rather
than the model.  For example, the Jacobian of the deterministic
transformation $h$ in the reversible-jump MCMC framework
\citep{green-2009-reversible} is due to interpreting $h$ as a change
of variables.  One must therefore compute this density
correction with respect to the proper base measure whenever jumping
between submanifolds of $\R^n$.  For example, such jumps would occur
in model selection involving a model with a latent von
Mises-Fisher-distributed random variable (which is defined on the unit
sphere).

The same problem can \emph{also} show up without any changes of
variables at all.  It suffices to compare the density of different
points under the same distribution, or the density of the same point
under different distributions.  For instance, the Indian GPA problem discussed in e.g. \cite{russell2018mixtures}
boils down to treating a density under counting measure as
comparable to a density under Lebesgue measure on $\R$, even though the
former represents an infinitely larger mass of probability.
This incarnation of the base measure problem is more widely studied
in the literature, albeit without a name.  Our solution smoothly
covers this scenario, as we will see in Section~\ref{sec:many-points}.

\section{What's the right answer?}
\label{sec:solution}

The root of the base measure problem is that we didn't want to compute with
measures directly, but lost the base measure when representing
probability distributions with densities.  It's not actually possible
to infer the correct base measure from the data type representing the
sample: a point on the unit circle in $\R^2$ is represented with two
floating-point numbers, but using Lebesgue measure on $\R^2$ as the
base is not helpful.

We propose standardizing on Hausdorff measure as a universal base
measure for probability densities.  Specifically, for $P$ a
distribution on a $d$-dimensional manifold $\M$ embedded in $\R^n$,
let's define the density to be with respect to the $d$-dimensional
Hausdorff measure $H^d$.  We propose $d$-Hausdorff measure because it
provides a coherent definition for ``$d$-dimensional volume'' of a
surface embedded in $\R^n$, even when the surface is
curved.\footnote{We propose standardizing at all in order to minimize
  the amount of information carried by the base measure, so it can be
  as implicit as possible.  The only degree of freedom in Hausdorff
  measure is the dimension, and that is unavoidable.} We scale $H^d$
to agree with the standard volume of straight surfaces, and thus with
Lebesgue measure when $d = n$.  We begin with fixed-dimension
distributions for simplicity, and briefly address mixed-dimension
distributions in Section~\ref{sec:other-supports}.

\subsection{Manifolds of dimension $d$}
\label{sec:single-dimension}

What do we need to represent about a sample $x$ besides
its density $p(x)$ in order to compute in this scheme?  Clearly we need to represent $d$; but, as
the example of the unit circle shows, that's not enough to transform
densities correctly.  Recall that our scaling transformation had a
constant Jacobian when viewed as a function from $\R^2$ to $\R^2$, but
we have to apply a position-dependent correction to account for its
effect on the unit circle density.  The additional thing we need to
represent is the tangent space of $\M$ at the sample point $x$.  Notably,
we do not need to computationally represent any non-local information
about $\M$ or $P$---just the tangent space at $x$, its dimensionality
$d$, and the density $p(x)$.

\begin{figure*}[tbp]
  \centering
\begin{tikzpicture}[scale=0.4]
  \draw plot [smooth] coordinates {(0,5) (3,5) (7,4) (10,0)};
  \draw plot [smooth] coordinates {(10,0) (13,2) (16,2.5)};
  \draw plot [smooth] coordinates {(16,2.5) (9,7) (0,5)};
  \node at (2,4) {$\mathcal{M}$};
  \draw[->] (17,5) -- (21,5);
  \node at (19,6) {$f$};
  \draw plot [smooth] coordinates {(22,4) (25,5.5) (29,8)};
  \draw plot [smooth] coordinates {(29,8) (30,5.5) (34,3)};
  \draw plot [smooth] coordinates {(34,3) (27,1.5) (22,4)};
  \node at (33.5,2) {$f\mathcal{M}$};
  \node[circle,fill,inner sep=1pt,label=above:$x$] at (10.5,4) {};
  \draw[->] (11,2.5) -- (12.5,3.8);
  \node at (12.5,2.85) {$\delta v_0$};
  \draw[dotted] (12.5,3.8) -- (10,5.5);
  \draw[->] (11,2.5) -- (8.5,4.2);
  \node at (9.6,2.9) {$\delta v_1$};
  \draw[dotted] (8.5,4.2) -- (10,5.5);
  \node[circle,fill,inner sep=1pt,label=above left:$f(x)$] at (28,3.5) {};
  \draw[->] (27.5,0.5) -- (29,5);
  \node at (29,2.8) {$\delta v_0'$};
  \draw[dotted] (29,5) -- (28.5,6.5);
  \draw[->] (27.5,0.5) -- (27,2);
  \node at (26.5,0.5) {$\delta v_1'$};
  \draw[dotted] (27,2) -- (28.5,6.5);
\end{tikzpicture}
  \caption{Change of local
    2-volume under a transformation from $\R^3$ to $\R^3$.  On the
    left, the small parallelogram with sides $\delta v_0$ and $\delta
    v_1$ is tangent to $\M$ at $x$ and forms a volume element.  On the
    right, $f$ takes $\M$ to $f\M$, $x$ to $f(x)$, and the volume
    element to the small parallelogram tangent to $f\M$ at $f(x)$ with
    sides $\delta v_0'$ and $\delta v_1'$.  When computing the density
    at $f(x)$ of the pushforward distribution $fP$ supported on $f\M$,
    we have to correct for the change in volume of this parallelogram,
    but we have to disregard any stretching or compression $f$
    may do in the perpendicular direction.  Theorem~\ref{thm:correction}
    formalizes this for arbitrary dimensions.}
  \label{fig:volume-elt}
\end{figure*}

Let us now formalize what we are actually trying to compute and how we
can compute it.

\begin{defn}
  Given a probability distribution $P$ over any space $\Omega$, and
  any function $f: \Omega \to \Omega'$, the \emph{pushforward}
  distribution $fP$ is the measure
  \[ fP(S) = P(f^{-1}(S)) \]
  for all events $S \subset \Omega'$ for which $f^{-1}(S)$ is
  measurable in $\Omega$.
\end{defn}

The sampler for $fP$ is just to apply $f$ to a sample from $P$.
We seek a tractable density function for $fP$, and fortunately
the following theorem gives it to us for nice $f$.

\begin{theorem}
  Consider
  \begin{itemize}
  \item
    A probability distribution $P$ over a $d$-dimensional manifold $\M$ in $\R^n$,
  \item
    with density $p$ with respect to $d$-dimensional Hausdorff measure
    $H^d$;
  \item a diffeomorphism $f: \R^n \to (f\R^n \subseteq \R^m)$; and
  \item a point $x \in \M$.
  \item Let $v_i$ be an arbitrary basis for the tangent space to $\M$ at $x$, and
    pack it into a $d$-by-$n$ matrix $V$.
  \end{itemize}
  Then
  \begin{enumerate}
  \item The pushforward $fP$ is supported on the pushforward manifold
    $f\M$ in $\R^m$.
  \item The directional derivatives $v_i' = d f(x + \eps v_i) / d \eps$
    are a basis for the tangent space to $f\M$ at $f(x)$.  Let us
    pack them into a $d$-by-$m$ matrix $V'$.
  \item The density $fp$ of $fP$ with respect to $H^d$ at $x' = f(x)$
    is
    \begin{equation}
      fp(x') = p(x)\frac{\sqrt{\det(VV^T)}}{\sqrt{\det(V'V'^T)}}.
      \label{eqn:correction}
    \end{equation}
  \end{enumerate}
  \label{thm:correction}
\end{theorem}

Claims 1 and 2 are standard, and give us the rule for propagating a
basis for the tangent space at $x$ to $x'$.  Claim 3 amounts to a
restatement in probability terms of standard notions of change of
volume.  The intuition is that the
relevant volume is volume in the tangent space to $\M$ at $x$ (see Figure~\ref{fig:volume-elt}).  We
choose as a volume element a small parallelepiped: centered at $x$ and with sides
$\delta v_i$ for $\delta \to 0$.\footnote{Since the $v_i$ are vectors
  in the tangent space to $\M$ at $x$, it would be more natural to let
  $x$ be the vertex of the parallelepiped.  However, we want to end
  up with a neighborhood of $x$, so we shift the parallelepiped to have
  $x$ in the center.}
The volume of this element is
$\delta^d \sqrt{\det(VV^T)}$.  The function $f$ transforms this volume
element to the parallelepiped with center $x'$ and sides $\delta
v'_i$.  Therefore, the volume correction we are looking for is
$\sqrt{\det(V'V'^T)}/\sqrt{\det(VV^T)}$.
We reproduce a formal proof of Claim 3 in Appendix~\ref{sec:proof}.

As an aside, Theorem~\ref{thm:correction} applies without modification
to discrete spaces: $d = 0$, 0-Hausdorff measure is, up to scaling,
identical with the usual counting measure, and we are free to define
all functions of a discrete domain to be vacuously diffeomorphic.  The
matrices $V$ and $V'$ have zero rows, and we recover the standard
result that $fp(x') = p(x)$ for $p$ a probability mass function on a
discrete space.

Observe that the inputs required for Theorem~\ref{thm:correction} are
all local: the point $x$, a basis $v_i$ for the tangent space to $\M$
at $x$, the density $p(x)$ of $P$ at $x$, and the ability to compute
the value and directional derivatives of $f$ at $x$.  We do not need
any non-local information about $P,\ \M$, or $f$, except that $f$ is invertible.  Observe also that the
outputs given by Theorem~\ref{thm:correction} give enough information
about $fP$ at $x' = f(x)$ to apply the same theorem again with a new
diffeomorphism $g : \R^m \to (g\R^m \subseteq \R^k)$.
We invite the interested reader
to verify that pushing distributions forward by
Theorem~\ref{thm:correction} commutes with function composition, that
is, the double pushforward $g(fP)$ computes the same densities as
a single pushforward by the composition $(g\circ f)P$.  This equivalence
lends itself to modularity in software, permitting, for example, the
distribution $g(fP)$ to be represented taking advantage of special
structure in $f$ or $g$, without requiring the user to derive
that special structure for $g\circ f$.

\subsection{Mixed-dimension manifolds}
\label{sec:other-supports}

Theorem~\ref{thm:correction} has given us a formula for transforming
densities of probability distributions on $d$-dimensional manifolds
embedded in $\R^n$, and justifies standardizing on $d$-dimensional
Hausdorff measure as a universal base measure for densities of such
distributions.

The results generalize directly to finite unions of manifolds, of
potentially different dimensions.  Namely, let $P$ be a sum of measures
$P_j$ on manifolds $\M_j$ of dimensions $d_j$, where without loss
of generality we require that the intersection of any $\M_j$ and $\M_k$ have
strictly lower dimension than at least one of $d_j$ and $d_k$.
In this situation, we use
the sum of the corresponding Hausdorff measures $H^{d_j}$ on those
manifolds as the base measure.  Then the density $p(x)$ representing
$P$ is given by the density $p_j(x)$ of the $P_j$ corresponding to the
lowest-dimensional manifold $\M_j$ that $x$ is on.
If in addition to $\M_j$, $x$ is on manifolds of higher dimension, any
densities $p_{j'}$ corresponding to measures on those higher-dimensional
manifolds become zero with respect to $H^{d_j}$.

If $x$ occurs on
multiple support manifolds of minimal dimension, then by assumption
$x$ is part of a lower-dimensional intersection, whose total mass
under $P$ is necessarily zero.  We are therefore formally free to
assign whatever probability density we wish to $x$.  We conjecture
that reporting the sum $\sum_{\{j | d_j \mathrm{minimal}\}} P_j(x)$
will follow the principle of least surprise, but are prepared to
be corrected by future experience in probabilistic programming.

\subsection{Multiple points}
\label{sec:many-points}

When doing probabilistic inference, we concern ourselves with more
than just one point $x$. We must compare the probability densities of
different points to assess which are the best (and by how much), and
we must add the probability densities of different points to form an
empirical mixture.  In MCMC we do the former when computing the
acceptance ratio; and in SMC and likelihood weighting we do both when
we resample or form the self-normalizing estimate of the evidence.
These comparisons and additions are a priori vulnerable to the base
measure problem; so to fix them, let us derive measure-aware formulas.

When computing a probability density ratio $p(y) / p(x)$,
we are actually computing the limiting ratio of finite probability masses
\begin{align}
  \Rat_P(y, x) = \lim_{\delta \to 0} P(B^\delta_y) / P(B^\delta_x),
  \label{eqn:ratio-limit}
\end{align}
where $B^\delta_x$ denotes the radius-$\delta$ ball about
$x$.\footnote{For this to be well-defined, $x$ and $y$ must live in a
single metric space.  For purposes of computational probability, it
suffices to assume we are operating in a disjoint sum of as many
copies of $\R^n$, for as many values of $n$, as needed.  We use the
Euclidean distance as a metric within each $\R^n$, and posit that the
distance between points in different $\R^n$ is infinite.} As $\delta$
becomes small, $P(B^\delta_x)$ becomes $C p(x) \delta^{d_x}$, where
$d_x$ is the (Hausdorff) dimension of the support of $P$ at $x$,
and $C$ is a constant that depends on the metric but not on $\delta$ or $x$.  Our
limit is therefore
\begin{align}
  \Rat_P(y, x) = \lim_{\delta \to 0} \delta^{d_y - d_x} p(y) / p(x).
  \label{eqn:ratio}
\end{align}

In the common case where $d_x = d_y$, this of course reduces to $p(y) / p(x)$; but
in the general case, (\ref{eqn:ratio}) becomes
\[ \Rat_P(y, x) = \begin{cases}
  p(y) / p(x) & \mathrm{when}\ d_y = d_x; \\
  0 & \mathrm{when}\ d_y > d_x, p(y) < \infty, \\
    & \mathrm{and}\ p(x) > 0; \\
  \infty & \mathrm{when}\ d_y < d_x, p(x) < \infty, \\
    & \mathrm{and}\ p(y) > 0; \\
  \mathrm{undefined} & \mathrm{otherwise}.
\end{cases}
\]
Note that in the ``otherwise'' case, the limit (\ref{eqn:ratio-limit})
may well exist, but the quantities $d_x$, $d_y$, $p(x)$, and $p(y)$ do
not suffice to determine it; so for computational purposes, we
throw up our hands and say ``undefined''.

Summation is more subtle, because it is not dimensionless.
Let us define the two-output function
\begin{align}
  \Summ_P(x, y) = d, \lim_{\delta \to 0} \frac{P(B_x^\delta) + P(B_y^\delta)}{C\delta^d},
  \label{eqn:sum-limit}
\end{align}
where $d$ is chosen to make the limit finite and positive.\footnote{This $d$ is unique when it exists.}
Then, using the same notation for $d_x$ and $d_y$, we derive
\[ \Summ_P(y, x) = \begin{cases}
  d_x, p(x) + p(y) & \mathrm{when}\ d_x = d_y; \\
  d_x, p(x) & \mathrm{when}\ d_x < d_y, p(y) < \infty, \\
    & \mathrm{and}\ p(x) > 0; \\
  d_y, p(y) & \mathrm{when}\ d_x > d_y, p(x) < \infty, \\
    & \mathrm{and}\ p(y) > 0; \\
  \mathrm{undefined} & \mathrm{otherwise}.
\end{cases}
\]
Again, we let the ``otherwise'' case be ``undefined'', because while
there may still be a $d$ for which the limit in (\ref{eqn:sum-limit})
exists, we cannot compute it from just $d_x$, $d_y$, $p(x)$, and
$p(y)$.

For completeness, we mention multiplying probabilities, even though it
is straightforward:
\[ \Prodd_P(x, y) = d_x + d_y, p(x) p(y). \]

Note that to compute any of $\Rat_P$, $\Summ_P$, or $\Prodd_P$, the
only information we need about the measure $P$ is local: the density
$p(\cdot)$ and support dimension $d_{(\cdot)}$ at the points $x$ and
$y$.  The support dimension $d$ tells us which Hausdorff measure $H^d$
serves as the base measure for the density $p(\cdot)$, and that's
enough to add, multiply, and divide, even when $P$ spans manifolds of
many dimensionalities in different places.  This $d$ is the number of
basis vectors spaning the tangent space required to compute volume
corrections per Section~\ref{sec:single-dimension}; but for $\Rat_P$,
$\Summ_P$, and $\Prodd_P$ we do not actually need the basis itself.

\section{How do we fix the software?}
\label{sec:software}

We discuss how to adjust current probabilistic programming software to
fix the base measure problem.  The largest change is needed to enable
correct changes of variables (Section~\ref{sec:change-of-vars}); given
that, updating common inference algorithms is a local adjustment, which
we explicate for completeness in Section~\ref{sec:inference-algs}.

\subsection{Correct changes of variables}
\label{sec:change-of-vars}

Correctly updating densities under changes of variables requires two
new architectural features of probabilistic programming software:
First, querying a probability distribution $P$ at a point $x$ must
produce a local measure, which includes the density $p(x)$, the local
dimension $d_x$, and a representation of the tangent space at $x$ to
the support manifold $\M$.
Second, changes of variable induced by deterministic transformations
must take the tangent space into account.
The current standard architecture can be seen as the specialization of
the above to the case when the dimension $d$ is fixed, and the support
manifold $\M$ is the entire host space $\R^n$.

To elaborate, we know from Section~\ref{sec:single-dimension} that to compute
the density of a pushforward $fP$ at a point $x'$ with respect to
$d$-Hausdorff measure it suffices to compute
\begin{itemize}
\item The preimage point $x = f^{-1}(x')$;
\item The density $p(x)$ of $P$ at $x$;
\item A basis $v_i$ for the tangent space to the support $\M$ of $P$ at $x$;
\item The directional derivatives $v_i'$ of $f$ at $x$ in directions
  $v_i$, forming a basis for the tangent space at $x'$ to the support
  of $fP$; and
\item The correction term
  $\sqrt{\det(VV^T)}/\sqrt{\det(V'V'^T)}$ from
  (\ref{eqn:correction}).
\item For composition, we also need to return the pushforward tangent basis $v_i'$
\end{itemize}

At first glance this may seem prohibitively expensive: even if we have
an automatic differentiation system ready to hand that lets us compute
the needful directional derivatives of $f$, do we really have to do
that, and then perform two full matrix multiplies and determinants,
for every probability density evaluation?

Fortunately, there are several special cases we can take advantage of
to save work, and to recover the performance of the conventional
architecture when it gives the correct answer:

\begin{itemize}
\item If our distribution $P$ is actually discrete, the tangent space
  is 0-dimensional, $VV^T$ is the 0x0 matrix, and its
  determinant is vacuously 1.
\item If $P$ is supported on all of $\R^n$, the tangent space is also
  $\R^n$, and we may take $v_i$ to be the standard basis, so the
  $\sqrt{\det(VV^T)}$ term is unity and we need not explicitly compute it.
\item If in addition $f$ maps $\R^n$ to $\R^n$ and not to $\R^m$ for
  $m > n$, then the matrix $V'$ is the Jacobian of $f$ at $x$ and
  itself has a determinant.  We can save a matrix multiply and a
  square root because in this case $\sqrt{\det(V'V'^T)} = |\det V'|$.
  This is where we recover the conventional Jacobian-determinant correction to
  probability densities.
\item Distributions on structured-sparse matrices, such as
  lower-triangular matrices, are supported on manifolds with
  axis-aligned tangent spaces.  In such a case, we need only keep
  track of a mask representing the present dimensions, taking the
  basis $v_i$ to be a subset of the standard basis.
  The $\sqrt{\det(VV^T)}$ term again disappears, though if $f$ does
  not preserve the sparsity, the $\sqrt{\det(V'V'^T)}$ term may need to be
  computed.
\item If in addition $f$ is a univariate transformation applied
  coordinate-wise, the sparse structure will be preserved, and the
  correction will just be the product of partial derivatives of $f$,
  $\sqrt{\det(V'V'^T)} = \prod_i \frac{\partial}{\partial x_i}f(x)$,
  with $i$ here ranging over the dimensions present in the manifold
  but not the others.
\item Finally, simplices and symmetric matrices may also merit special
  treatment, because the basis $v_i$ can again be implicit.
\end{itemize}

This list of special cases suggests ad-hoc polymorphism as a
representation strategy.  Specifically, a typical probabilistic
programming system probably already has a Bijector class (hierarchy)
for representing transformations $f$.  To implement this strategy for
volume corrections, we can add a TangentSpace class hierarchy for the
tangent space to $\M$ at $x$.  This second hierarchy can have
dedicated classes for efficient cases, such as the zero space
tangent to a discrete support, the full space tangent to
an $n$-dimensional support in $\R^n$, and so on.  To take
advantage of special cases based both on the tangent space and the
bijector, we can define the density correction function with
two-argument dispatch, or a Visitor-like pattern \citep{gang-of-four}
to emulate it in languages where only single-argument dispatch is
available.

When working with distributions on high-dimensional spaces, or with
large batches of distributions, we expect the cost of the dispatch to
be small compared with the cost of the linear algebra the dispatch
avoids.  Further, on a tracing platform like JAX~\citep{jax},
TensorFlow~\citep{tensorflow2015-whitepaper}, or
TorchScript~\citep{torchscript}, the dispatch only occurs once during
tracing, avoiding repeated linear algebra during execution.

\subsection{Measure-aware inference algorithms}
\label{sec:inference-algs}

Once a probabilistic programming system correctly propagates local
measures per Section~\ref{sec:change-of-vars}, defending standard
inference algorithms against the base measure problem is just a matter
of using the formulas from Section~\ref{sec:many-points} for
operations that compare or combine probability densities.  We note
again that only the density and local dimension are needed here;
the tangent basis is only used for changing variables.

We spell out two standard inference algorithms explicitly.  To
resample in SMC, just throw out particles of non-minimally-dimensioned
weights, as in Algorithm~\ref{alg:smc}.  In the one-stage case this
reduces to the dimensionality-aware likelihood weighting algorithm of
\cite{russell2018mixtures}.
\begin{algorithm}[h]
  \caption{Measure-aware resampling for SMC}
  \label{alg:smc}
  \begin{algorithmic}
    \STATE {\bfseries Input:} Particles $p_j$, dimensioned weights $(d_j, w_j)$
    \STATE {\bfseries Input:} Desired number of new particles $N$.
    \STATE Compute $d = \min d_j$.
    \IF{All $w_j = 0$ where $d_j = d$}
    \STATE {\bfseries Error:} The resampling is undefined.
    \ENDIF
    \IF{Any $w_j = +\infty$ where $d_j > d$}
    \STATE {\bfseries Error:} The resampling is undefined.
    \ENDIF
    \FOR{$i=1$ {\bfseries to} $N$}
    \STATE Sample index $k_i \in \{j | d_j = d\}$ proportional to $w_j$.
    \STATE {\bfseries Output} $p_{k_i}$.
    \ENDFOR
    \STATE {\bfseries Return} Dimensioned self-normalization estimate $(d, \sum_{j | d_j = d} w_j)$.
  \end{algorithmic}
\end{algorithm}

To compute the Metropolis-Hastings acceptance or rejection for MCMC,
just compare the target density dimension-major, as in
Algorithm~\ref{alg:mh}.  Note that the proposed state $y$ must live on
a same-dimension manifold as the current state $x$, otherwise the
proposal $q$ is not reversible.  This can be arranged with the
usual reversible-jump MCMC technique from \cite{green-2009-reversible}.
However, for general probabilitic programs $\pi$, $x$ and $y$ may yield
differently-dimensioned likelihoods, and M-H must account for this.
\begin{algorithm}[h]
  \caption{Measure-aware Metropolis-Hastings}
  \label{alg:mh}
  \begin{algorithmic}
    \STATE {\bfseries Input:} Target $\pi(\cdot)$, proposal $q(\cdot|\cdot)$, current state $x$
    \STATE Sample proposed state $y \sim q(\cdot|x)$.
    \STATE Evaluate dimensioned probabilities
    \STATE $(d_x, p_x) = \pi(x)$
    \STATE $(d_y, p_y) = \pi(y)$
    \IF{$d_y > d_x$}
    \IF{$p_x \neq 0$ and $p_y < +\infty$}
    \STATE {\bfseries Reject:} $d_x, p_x$ dominates $d_y, p_y$.
    \ELSE
    \STATE {\bfseries Error:} The comparison is undefined.
    \ENDIF
    \ENDIF
    \IF{$d_y < d_x$}
    \IF{$p_y \neq 0$ and $p_x < +\infty$}
    \STATE {\bfseries Accept:} $d_y, p_y$ dominates $d_x, p_x$.
    \ELSE
    \STATE {\bfseries Error:} The comparison is undefined.
    \ENDIF
    \ENDIF
    \STATE Sample $u$ uniformly between 0 and 1.
    \STATE {\bfseries Accept} if $u < \frac{p_y q(x|y)}{p_x q(y|x)}$, else {\bfseries Reject}.
  \end{algorithmic}
\end{algorithm}

For example, in the much-discussed Indian GPA problem (see
\cite{russell2018mixtures} for one restatement), we might do MCMC on a
boolean latent variable modeling whether a student is American or
Indian.  The dimensioned likelihood of the observed GPA of 4.0 for an
American student would be something like $(0, 0.1)$, indicating a
point-mass of size 0.1 at the maximum GPA that American high schools
report.  Similarly, the dimensioned likelihood of a 4.0 GPA for an
Indian student would be $(1, 0.1)$, being the density of a scalar
continuous uniform distribution of GPAs between 0 and 10.
Algorithm~\ref{alg:mh} would then correctly accept every proposal to
change the latent from ``Indian'' to ``Ameican'', and reject every
proposal to change back.

\section{Related Work}

Broadly, all general-purpose probabilistic programming systems are
related to the present work, in that they must either produce
incorrect results when the base measure problem arises
\citep{quantized-distribution-bug, venture-base-measure-problem}, avoid
the problem entirely, or somehow solve it.  Three specific systems are
worth mentioning:

Stan \citep{carpenter2017stan} is an example of a system that avoids
the base measure problem.  There are two places in a Stan model where
the base measure problem might occur: in the automatic unconstraining
transformations and in applying densities to transformed parameters.
The former Stan can get right because both the base measure and
unconstraining transformation are implied by the type of the parameter
being unconstrained.  The base measure problem occurs, in a sense, but
the remedy is both local and internal to the implementation of Stan.
The latter Stan pushes on the user: if a user writes a Stan model
with a parameter $x$ and a transformed parameter $f(x)$, and wishes
to code a density on $f(x)$ that corresponds to the pushforward of
some known density $p(x)$, then that user must think about base measures
and partial derivatives themselves.  Stan does not help; but then
again, it also does not claim it would help, so Stan doesn't compute
anything visibly ``incorrect''.

BLOG discusses addressing the base measure problem (not by that name)
in \cite{russell2018mixtures}.  Relative to the present work, that
treatment is simpler in two ways: they only discuss scalar random
variables, where every tangent space is either $\R$ or the zero vector
space; and they do not explicitly discuss transformations of random
variables.

Some cases of the base measure problem are addressed by Hakaru's
disintegration transform \citep{narayanan-shan-2020}, which generalizes
density and also needs correct base measures. Hakaru defines a
restricted language of base measures with respect to which it can
symbolically compute (unnormalized) disintegrations of s-finite
measures. As this restricted set includes discrete-continuous
mixtures, it is sufficient to correctly handle the Indian GPA problem;
but for the present example of the unit circle, it would
need a richer notion of constraints and their tangent or normal spaces \citep{hakaru-base-measure-problem}.

For example, the following model of the ellipse example from
Section~\ref{sec:example} is expressible in the system from
\citep{narayanan-shan-2020} but gives an incorrect disintegration:
\begin{verbatim}
x ~ lebesgue
y ~ lebesgue
observe x**2 + y**2 = 1
return (2*x, 20*y)
\end{verbatim}
The issue is that the constraint restricts the distribution to a
lower-dimensional manifold $\M$, and the tangent space to that
manifold is not taken into account when applying the Jacobian
correction due to the transformation $f(x, y) = (2x, 20y)$.  This
tangent space could be recovered, since it is perpendicular to the
gradient of the constraint, so we hope that this manifestation of
the base measure problem will be relatively easy to fix in Hakaru.

There is also a line of work on formal semantics of probabilistic
programming languages, for example \cite{borgstrom-etal-2013} and
\cite{staton-etal-2016}.  That work
taken on its own terms is generally not vulnerable to the base measure
problem at all, because the semantics are invariably in given terms of
measures, not densities.  The present work can be viewed as a bridge
to implementation, providing a complete and efficient local representation of a measure.

\section{Future Work}

A common thread in the above related work is to address base measures
dimension-wise.  In other words, reduce to the scalar-oriented setting
and track one ``discrete or $\R$'' bit per scalar random variable.
This is appealing because no explicit tangent spaces are needed, since
they are all either trivial or $\R$.  A vector-valued distribution $P$
on a $d$-dimensional manifold embedded in $\R^n$ can then be modeled
as $n$ deterministic scalar functions $f_i$ from $d$ scalar random
variables.  Computing the density of $P$ involves constraining the
outputs of the $f_i$ one by one, with a density correction given by
the derivative of $f_i$ as appropriate.  After $d$ irredundant
constraints, the remaining distribution is discrete, and no further
derivatives are needed.

It would be interesting to work out whether
the preceding strategy actually works.  The difficulty we predict is that,
even if a joint function $f$
is invertible as a vector-valued function, its coordinates $f_i$ may
not be invertible individually.  In this case, treating them separately
seems to require explicitly manipulating their preimages, which may be
non-trivial discrete sets.  For example, if we represent a
distribution on the unit circle as a distribution on $[0, 2\pi)$
  followed by the functions $(\sin \cdot, \cos \cdot)$ then the
  preimage of $\sin$ doesn't actually give us a unique point at which
  to evaluate its derivative; whereas the preimage of $\cos$, while no
  longer contributing a derivative per se, is still needed to identify
  the correct branch of $\sin$.

\section{Conclusion}

We have named the \emph{Base Measure Problem} and provided a solution
to it.  Implementing the solution in probabilistic programming systems
should cause negligible loss of performance for cases that were
already correctly handled, and expand the set of models in which the
system can compute correct probability densities.  Implementation does
carry a code complexity cost, but that cost is minimized by using
two-argument dispatch, or emulating it with a Visitor pattern.
Despite correctly accounting for measures, no non-local information is
required.

\section{Acknowledgments}


The authors thank Srinivas Vasudevan, Wynn Vonnegut, and Praveen Narayanan for
enlightening discussion.  The authors also thank the anonymous
reviewers for valuable constructive criticism.

\bibliographystyle{apalike}
\bibliography{main}

\appendix

\section{Proof of Theorem~\ref{thm:correction}, Claim 3}
\label{sec:proof}

To prove Theorem~\ref{thm:correction}, Claim 3 formally, we go through the measures $P$ and $fP$,
starting with the following lemma.  The only technical trick is to
enlarge our parallelepipeds to open sets in $\R^n$ and $\R^m$, so that
measuring them with $P$ and $fP$ captures all the mass near $x$ and
$x'$, respectively, despite any curvature of $\M$ or $f\M$.  Then we
will note that in the limit the only thing we care about is the
projections of those open sets back to the respective tangent spaces,
which are the parallelepipeds we started with.

\begin{lemma}
  Let $P$, $\M$, $p$, $x$, $v_i$, and $V$ be as in the statement of Theorem~\ref{thm:correction}.  Let
  $G_V$ be the $d$-parallelepiped with center $x$ and sides
  $v_i$, and let $G_{V+}$ be any $n$-parallelepiped centered at $x$ with section $G_V$.
  Let $G_V^\delta$ and $\neigh$ be $G_V$ and $G_{V+}$, respectively, scaled by $\delta$ about $x$.
  Then the density $p$ of $P$ with respect
  to $d$-dimensional Hausdorff measure is given by
  \[ p(x) = \frac{1}{\sqrt{\det(VV^T)}}
     \lim_{\delta \to 0} P(\neigh) / \delta^d. \]
\end{lemma}

\begin{proof}[Proof of lemma]
The definition of $p$ being a density for $P$ with
respect to $H^d$ is that for all measurable subsets $B \subset \R^n$,
\[ P(B) = \int_B p\, dH^d. \]
$P$ is only supported on $\M$, so we may take the integral on the right
to be over the intersection $B \cap \M$.  We choose for $B$ the
neighborhood $\neigh$.

As $\neigh$ becomes sufficiently small, we may assume $p$ is
constant on $\neigh \cap \M$, so
\[ P(G^\delta_{V+}) \approx p(x) H^d(G^\delta_{V+} \cap \M). \]
As $\delta \to 0$, $\M$ near $x$ approaches its tangent space, so
$G^\delta_{V+} \cap \M$ becomes $G^\delta_V$.  The right-hand term
becomes the $d$-dimensional Hausdorff measure of $G^\delta_V$, which
given our choice of scaling for the Hausdorff measure is just the
volume thereof, namely $\delta^d \sqrt{\det(VV^T)}$.
\end{proof}

\begin{proof}[Proof of Claim 3]
  Let $G_{V'}$ be the $d$-parallelepiped in $\R^m$ with center
  $x'$ and sides $v'_i$.  Let $G_{V'+}$ be any $m$-parallelepiped in $\R^m$
  with section $G_{V'}$ centered at $x'$.  Let $\neighp$ be $G_{V'+}$ scaled by $\delta$ about $x'$.  By the lemma, we have
  \[ fp(x') = \frac{1}{\sqrt{\det(V'V'^T)}} \lim_{\delta \to 0} fP(\neighp) / \delta^d. \]
  The preimage $f^{-1}(\neighp)$ approaches an $n$-parallelepiped
  centered at $x$ with section $G^\delta_V$.  This is because the basis vectors $v_i'$ forming
  $V'$ are the directional derivatives of $f$ in the directions given
  by the basis vectors $v_i$ forming $V$.  In directions normal to $\M$, we just need the
  derivatives of $f$ to be finite.
  We may thus apply the lemma again to write
  \begin{align*}
    \lim_{\delta \to 0} fP(\neighp) / \delta^d
    & = \lim_{\delta \to 0} P(f^{-1}(\neighp)) / \delta^d \\
    & = \sqrt{\det(VV^T)} p(x),
  \end{align*}
  giving the desired result.
\end{proof}

\end{document}